\RequirePackage[l2tabu, orthodox]{nag}
\documentclass{cccg10}
\usepackage{etex,cite,booktabs,subcaption,amssymb,tikz,multirow,algorithmic}
\usepackage[textsize=small]{todonotes}
\usepackage{microtype}

\usetikzlibrary{shapes}
\usetikzlibrary{calc}

\newtheorem{corollary}[theorem]{Corollary}
\newtheorem{observation}[theorem]{Observation}

\newcommand{\norm}[1]{\left\Vert#1\right\Vert}

\newcommand{\eps}{\varepsilon}

\newcommand{\Frechet}{Fr\'echet }

\makeatletter
\pgfmathdeclarefunction{atan3}{2}{%
	\begingroup
		\pgfmathsetmacro \pgfmathx {abs(#1)}%
		\pgfmathsetmacro \pgfmathy {abs(#2)}%
		\ifdim \pgfmathx pt>\pgfmathy pt\relax%
			\pgfmathparse{atan2(#1,#2)}%
		\else
			\pgfmathparse{90 - atan2(#2,#1)}%
		\fi%
 		\pgfmath@smuggleone\pgfmathresult%
	\endgroup
}
\makeatother

\newcommand{\ceps}{0.2cm}

\newcommand{\tikzdefines}[0] {
    \tikzstyle{cblack}=[circle, draw, thick, solid, fill=black, scale=.15]
    \tikzstyle{cblue}=[circle, draw, solid, fill=cyan!20, scale=.4]
}
\pgfdeclarelayer{background2}
\pgfdeclarelayer{background}
\pgfsetlayers{background2,background,main}

\newcommand{\corner}[5] {
    \coordinate (a1) at ($(#3)!2*\ceps!270:(#2)$);
    \coordinate (a2) at ($(a1) + ($(0,0)!\ceps!($(#3)-(#2)$)$)$);
    \coordinate (a3) at ($(#4)!2*\ceps!90:(#5)$);
    \coordinate (a4) at ($(a3)+($(0,0)!\ceps!($(#4)-(#5)$)$)$);
    \coordinate (#1_b) at (intersection of a1--a2 and a3--a4);
    \coordinate (a1) at ($(#1_b)+($(0,0)!\ceps!($(#3)-(#2)$)$)$);
    \coordinate (a2) at ($(#1_b)+($(0,0)!\ceps!($(#5)-(#4)$)$)$);
    \coordinate (#1_a) at (intersection of #2--#3 and a2--#1_b);
    \coordinate (#1_c) at (intersection of #4--#5 and a1--#1_b);
    \coordinate (#1_d) at ($(#1_a)!\ceps!270:(#1_b)$);
    \coordinate (#1_e) at ($(#1_c)!\ceps!90:(#1_b)$);
    \coordinate (a1) at ($(#3)!2*\ceps!90:(#2)$);
    \coordinate (a2) at ($(a1) + ($(0,0)!\ceps!($(#3)-(#2)$)$)$);
    \coordinate (a3) at ($(#4)!2*\ceps!270:(#5)$);
    \coordinate (a4) at ($(a3)+($(0,0)!\ceps!($(#4)-(#5)$)$)$);
    \coordinate (#1_f) at (intersection of a3--a4 and #1_d--#1_e);
    \coordinate (#1_g) at (intersection of a1--a2 and #1_d--#1_e);
    
    \node[cblue] (#1_ne) at (#1_e) {};
    \node[cblue] (#1_nd) at (#1_d) {};
    \node[cblue] (#1_nf) at (#1_f) {};
    \node[cblue] (#1_ng) at (#1_g) {};
    \node[cblack] () at (#1_a) {};
    \node[cblack] () at (#1_b) {};
    \node[cblack] () at (#1_c) {};
    \node[cblack] () at (#1_f) {};
    \node[cblack] () at (#1_g) {};
}
\newcommand{\cornercap}[4] {
    \coordinate (#1_a) at ($(#3)!-#4!(#2)$);
    \coordinate (#1_b) at ($(#1_a)!2*\ceps!90:(#3)$);
    \coordinate (#1_c) at ($(#1_b)!2*#4!90:(#1_a)$);
    \node[cblue] (#1_d) at ($(#1_a)!-1*\ceps!90:(#3)$) {};
    \node[cblue] () at (#1_b) {};
    \node[cblack] () at (#1_a) {};
    \node[cblack] () at (#1_b) {};
}
\newcommand{\pathborder}[2]{
    [
        create hullcoords/.code={
            \global\edef\namelist{#1}
            \foreach [count=\counter] \nodename in \namelist {
                \global\edef\numberofnodes{\counter}
                \coordinate (hullcoord\counter) at (\nodename);
            }
            \pgfmathtruncatemacro \numberofnodes {\numberofnodes - 1}
            \foreach [count=\counter] \nodenum in {\numberofnodes,...,1} {
                \pgfmathtruncatemacro \next {\counter + \numberofnodes + 1}
                \coordinate (hullcoord\next) at (hullcoord\nodenum);
            }
            \pgfmathtruncatemacro \numberofnodes {\numberofnodes * 2 + 1}
            \coordinate (hullcoord0) at (hullcoord2);
            \pgfmathtruncatemacro \lastnumber {\numberofnodes+1}
            \coordinate (hullcoord\lastnumber) at (hullcoord2);
        },
        create hullcoords
    ]
    ($(hullcoord1)!#2!90:(hullcoord0)$) node[circle] {}
    \foreach [
        evaluate=\currentnode as \prevnode using \currentnode-1,
        evaluate=\currentnode as \nextnode using \currentnode+1
    ] \currentnode in {2,...,\numberofnodes} {
        let
            \p1 = ($(hullcoord\currentnode) - (hullcoord\prevnode)$),
            \n1 = {atan3(\x1,\y1) + 90pt},
            \p2 = ($(hullcoord\nextnode) - (hullcoord\currentnode)$),
            \n2 = {atan3(\x2,\y2) + 90pt},
            \n3 = {Mod(\n2-\n1,360) - 360},
            \n4 = {cos(.5*(\n3+360))},
            \n5 = {\n1+.5*\n3}
        in
        {
            \ifdim \n3 < -180.05pt {
                -- ($(hullcoord\currentnode) + (\n5:-1/\n4*#2)$)
            } \else {
                -- ($(hullcoord\currentnode)!#2!-90:(hullcoord\prevnode)$)
                arc [start angle=\n1, delta angle=\n3, radius=#2]
            } \fi
        }
    }
}

\title{Hardness Results on Curve/Point Set Matching with \Frechet Distance}

\author{Paul Accisano$^*$ \and Alper {\"{U}ng\"{o}r}
    \thanks{Dept. of Computer \& Info. Sci. \& Eng.,
        University of Florida, {\tt \{accisano, ungor\}@cise.ufl.edu}}
}

\begin{document}
\maketitle

\begin{abstract}
    Let $P$ be a polygonal curve in $\mathbb{R}^d$ of length $n$, and $S$ be a point-set of size $k$.  We consider the problem of finding a polygonal curve $Q$ on $S$ such that all points in $S$ are visited and the \Frechet distance from $P$ is less than a given $\eps$.  We show that this problem is NP-complete, regardless of whether or not points from $S$ are allowed be visited more than once.  However, we also show that if the problem instance satisfies certain restrictions, the problem is polynomial-time solvable, and we briefly outline an algorithm that computes $Q$.
\end{abstract}

\section{Introduction}
Measuring the similarity between two geometric objects is a fundamental problem in many fields of science and engineering.  However, to perform such comparisons, a good metric is required to formalize the intuitive concept of ``similarity.''  Among the many metrics that have been considered, \Frechet distance has emerged as a popular and powerful choice, especially when the geometric objects are curves.  Shape matching with \Frechet distance has been applied in many different fields, including handwriting recognition \cite{Sriraghavendra07}, protein structure alignment \cite{Jiang08}, and vehicle tracking \cite{Brakatsoulas05}.

In this paper, we consider the basic problem of measuring the similarity of two polygonal curves.  However, in our problem, the input is only partially defined.  Instead of being given both curves, we are given only one polygonal curve $P$ as well as a point set $S$.  Our problem is to complete this partial input by constructing a polygonal curve $Q$ that best matches the given curve, under the restriction that the constructed curve's vertices are exactly $S$.

Since our metric of choice is \Frechet distance, we begin with an popular, intuitive description the concept.  The metaphor of a person walking a dog is often used, with the dog walking along one curve and its owner walking along the other.  The \Frechet distance between the two curves is the length of the smallest leash that would allow both the person and the dog to reach the end of their respective curves without ever backtracking or letting go of the leash.  If a very short leash is sufficient, then the curves are very similar.  But if a longer leash is required, then the curves are very different.

Suppose our dog owner wants to walk his dog while walking down a path in a given park.  The extremely curious and territorial dog wants to sniff and/or mark every single tree along the path, running directly from tree to tree.  Mindful of their dog's proclivities, the owner goes shopping for a leash long enough to allow the dog to have its way without pulling the owner off the path.  For a given path and set of trees, will a leash of a given length be sufficient?  More formally, given a polygonal curve $P$, a point-set $S$, and a real number $\eps > 0$, does there exist a polygonal curve on $S$ that visits every point in $S$ and has \Frechet distance less than $\eps$ from $P$?  Unfortunately for owners of territorial dogs, we show in this paper this problem is NP-complete.

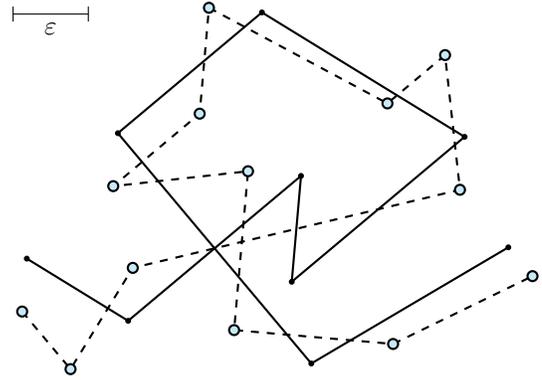
\begin{figure}
    \centering
    \begin{tikzpicture}[scale=.5]
    \tikzdefines
    \draw[thick, rotate around={40:(0,0)}] (-4,4) node[cblack] {}
        \foreach \pt in {(-3,1), (3,1), (1,-1), (7,-1), (5,5), (0,5), (0,-3), (6,-4)} {
            -- \pt node[cblack] {}
        };
    \draw[thick, rotate around={40:(0,0)}, dashed] (-5,3) node[cblue] {}
        \foreach \pt in {(-5,1), (-2,2), (6,-2), (8,1), (6,1), (4,6), (2,4), (-1,4), (2,2), (-1,-1), (2,-4), (6,-5)} {
            -- \pt node[cblue, thick] {}
        };
    
    \coordinate (a1) at (-6,7);
    \coordinate (a2) at ($(a1) + ($(0,0)!2cm!(1,0)$)$);
    \path (a1) edge node[below]{$\eps$} (a2);
    \draw ($(a1)+(0,-.2)$)--($(a1)+(0,.2)$);
    \draw ($(a2)+(0,-.2)$)--($(a2)+(0,.2)$);
\end{tikzpicture} 
    \caption{A problem instance and its solution.}
    \label{fig:example}
\end{figure}

\section{Previous Work and New Results}

The decision version of the \Frechet distance problem asks, given two geometric objects and a real number $\eps > 0$, is the \Frechet distance $\delta_F$ between the two objects less than $\eps$?  Alt and Godau \cite{Alt95} showed that, when the objects in question are polygonal curves of length $n$ and $m$, this problem can be solved in $O(nm)$ time.  They also showed that finding the exact \Frechet distance between the two curves can be done in $O(nm \log(nm))$ time.

Maheshwari et al.\ \cite{Maheshwari11} examined the following variant of the \Frechet distance problem, which we refer to as the Curve/Point Set Matching (CPSM) problem.  Given a polygonal curve $P$ of length $n$, a point set $S$ of size $k$, and a number $\eps > 0$, determine whether there exists a polygonal curve $Q$ on a subset of the points of $S$ such that $\delta_F(P, Q) \le \eps$. They gave an algorithm that decides this problem in time $O(nk^2)$.  They also showed that the curve of minimal \Frechet distance can be computed in time $O(nk^2\log(nk))$ using parametric search.

Wylie and Zhu \cite{Wylie12} also explored the CPSM problem from the perspective of discrete \Frechet distance.  In contrast to the continuous \Frechet distance, which takes into account the distance at all points along both curves, the discrete \Frechet distance only takes into account the distance at the vertices along the curves.  They formulated four versions of the CPSM problem depending on whether or not points in $S$ were allowed to be visited more than once (Unique vs. Non-unique) and whether or not $Q$ was required to visit all points in $S$ at least once (All-Points vs. Subset)  They showed that, under the discrete \Frechet distance metric, both non-unique versions were solvable in $O(nk)$ time, and both unique versions were NP-complete.

In this paper, we show that the Continuous All-Points versions of the CPSM problem, both Unique and Non-unique, are NP-complete.  Table \ref{tab:results} shows the eight versions of the problem, with our results highlighted.

\begin{table}[htbp]
    \centering
    \begin{tabular}{l@{\hspace{0.5em}}lc@{\hspace{0.3em}}lc}
        \toprule
                &               &   \multicolumn{2}{c}{Discrete} & Continuous\\
        \midrule
        Subset  & Unique        & NP-C & \cite{Wylie12}     & Open  \\
                & Non-Unique    & P    & \cite{Wylie12}     & P \cite{Maheshwari11} \\
        All-Pts & Unique        & NP-C & \cite{Wylie12}     & NP-C* \\
                & Non-Unique    & P    & \cite{Wylie12}     & NP-C* \\
        \bottomrule
    \end{tabular}%
    \caption{Eight versions of the CPSM problem and their complexity classes.  New results starred. \smallskip}
    \label{tab:results}%
\end{table}%

\section{Preliminaries}
Given two curves $P, Q : [0, 1] \rightarrow \mathbb{R}^d$, the \emph{\Frechet distance} between $P$ and $Q$ is defined as $\delta_F(P, Q) = \inf_{\sigma, \tau} \max_{t\in[0,1]} \norm{P(\sigma(t)), Q(\tau(t))}$, where $\sigma, \tau : [0, 1] \rightarrow [0, 1]$ range over all continuous non-decreasing surjective functions \cite{Ewing85}.  We make use of two commonly noted observations:

\begin{observation}
    Given four points $a$, $b$, $c$, $d \in \mathbb{R}^d$, if $\norm{ac} \le \eps$ and $\norm{bd} \le \eps$, then $\delta_F(\overrightarrow{ab}, \overrightarrow{cd}) \le \eps$.
\end{observation}

\begin{observation}
    Let $P_1$, $P_2$, $Q_1$, and $Q_2$ be four curves in $\mathbb{R}^d$ with $\delta_F(P_1, Q_1) \le \eps$ and $\delta_F(P_2, Q_2) \le \eps$.  If the ending point of $P_1$ (resp. $Q_1$) is the same as the starting point of $P_2$ (resp. $Q_2$) then $\delta_F(P_1 + P_2, Q_1 + Q_2) \le \eps$, where $+$ denotes concatenation.
\end{observation}

We now give a number of geometric definitions, some of which were used in \cite{Maheshwari11}.  For a given a point $p \in \mathbb{R}^d$ and a real number $\eps > 0$, let $\mathcal{B}(p, \eps) \equiv \{q \in \mathbb{R}^d : \norm{pq} \le \eps$ denote the \emph{ball} of radius $\eps$ centered at $p$, where $\norm{\cdot}$ denotes Euclidean distance.  For a line segment $L \subset \mathbb{R}^d$, let $\mathcal{C}(L, \eps) \equiv \bigcup_{p \in L} \mathcal{B}(p, \eps)$ denote the \emph{cylinder} of radius $\eps$ around $L$. Note that a necessary condition for two polygonal curves $P$ and $Q$ to have \Frechet distance less than $\eps$ is that the vertices of $Q$ must all lie within the cylinder of some segment of $P$.

Let the continuous function $P : [0, 1] \rightarrow \mathbb{R}^d$ represent a curve in $\mathbb{R}^d$.  Given two points $u, v \in P$, we use the notation $u \prec v$ if $u$ occurs before $v$ on a traversal of $P$.  The relation $\succ$ is defined analogously.  

\section{Restricted Satisfiability Problem}
Our NP-completeness result is obtained via reduction from a restricted version of the well-known 3SAT problem.  The 3SAT problem takes as input a boolean formula with clauses of size 3, and asks whether there exists an assignment to the variables that makes the formula evaluate to TRUE.  If we restrict the input to formulas in which each literal occurs exactly twice, the problem becomes the (3,B2)-SAT problem.  This may seem to be a rather extreme restriction, and, indeed, formulas of this type with less than 20 clauses are always satisfiable.  However, despite this restriction, the problem was shown to be NP-complete in \cite{Berman03}, and an example of an unsatisfiable formula with 20 clauses was presented.

In order to simplify our reduction, we make the further restriction that no two clauses have two literals in common.  In other words, we restrict the input to formulas in which the function that maps each literal to the pair of clauses it appears in is injective.  For any formula that violates this assumption, an equivalent, compliant formula can easily be constructed using the ``balanced enforcers'' described in \cite{Berman03}.  We therefore assume formulas to have this property for the remainder of the paper.

In the following sections, we first give a brief summary of the construction.  We then describe the the main gadget used, building it incrementally.  Finally, we give the full construction in detail.

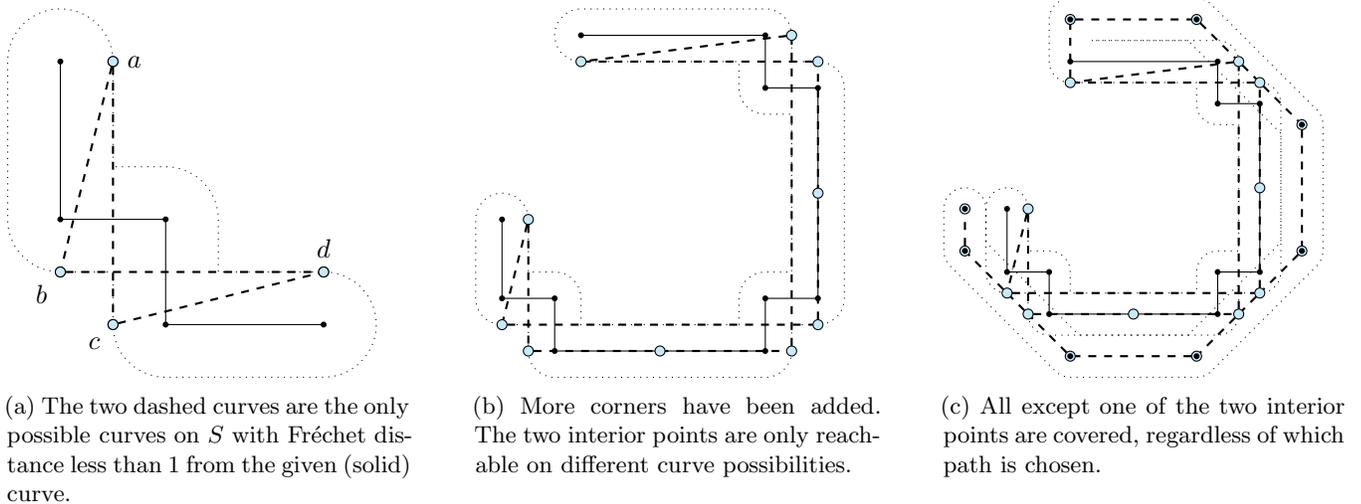
\begin{figure*}[ht]
    \begin{subfigure}[t]{0.3\textwidth}
        \begin{tikzpicture}[scale=.7,rotate=0]
    \tikzdefines

    \foreach [count=\x] \pt in {(0,5), (0,2), (2,2), (2,0), (5,0)}
        \node[cblack] (p\x) at \pt {};
    \draw (p1)--(p2)--(p3)--(p4)--(p5);
    \draw[dotted] \pathborder{p1,p2,p3,p4,p5}{1cm};

    \foreach [count=\x] \pt in {(1,5), (0,1), (1,0), (5,1)}
        \node[cblue] (s\x) at \pt {};

    \node[cblue,label=0:$a$] (s1) at (1,5) {};
    \node[cblue,label=-135:$b$] (s2) at (0,1) {};
    \node[cblue,label=-135:$c$] (s3) at (1,0) {};
    \node[cblue,label=90:$d$] (s4) at (5,1) {};

    \draw[thick, black, dashed] (s1) -- (s2) -- (s4);
    \draw[thick, black, dashed] (s1) -- (s3) -- (s4);
\end{tikzpicture} 
        \caption{The two dashed curves are the only possible curves on $S$ with \Frechet distance less than 1 from the given (solid) curve.}
        \label{fig:separation1}
    \end{subfigure}
    \hfill
    \begin{subfigure}[t]{0.3\textwidth}
        \begin{tikzpicture}[scale=.35]
    \tikzdefines

    \foreach [count=\x] \pt in {(0,5), (0,2), (2,2), (2,0), (10,0), (10,2), (12, 2), (12,10), (10,10), (10,12), (3,12)}
    	\node[cblack] (p\x) at \pt {};

    \draw ($(p1)$) \foreach \x in {2,...,11} {--($(p\x)$)};
    \draw[dotted] \pathborder{p1,p2,p3,p4,p5,p6,p7,p8,p9,p10,p11}{1cm};

    \foreach [count=\x] \pt in {(1,5), (0,1), (1,0), (11,0), (12,1), (12,11), (11,12), (3,11)}
    	\node[cblue] (s\x) at \pt {};

    \coordinate (i1) at ($(p4)!.5!(p5)$) {};
    \coordinate (i2) at ($(p7)!.5!(p8)$) {};

    \draw[thick, black, dashed] (s1) -- (s2) -- (s5) -- (s6) -- (s8);
    \draw[thick, black, dashed] (s1) -- (s3) -- (s4) -- (s7) -- (s8);
    \node[cblue] () at (i1) {};
    \node[cblue] () at (i2) {};

\end{tikzpicture} 
        \caption{More corners have been added.  The two interior points are only reachable on different curve possibilities.}
        \label{fig:separation2}
    \end{subfigure}
    \hfill
    \begin{subfigure}[t]{0.3\textwidth}
        \begin{tikzpicture}[scale=.28]
    \tikzdefines
    \renewcommand{\ceps}{1cm}

    \foreach [count=\x] \pt in {
            (1,5), (0,1), (1,0), (11,0), (12,1), (12,11), (11,12), (3,11)
            ,( 3,14), ( 9,14), (14, 9), (14, 3), ( 9,-2), (3, -2), (-2, 3), (-2, 5)
            }
     	\node[cblue] (s\x) at \pt {};

    \foreach [count=\x] \pt in {( 0, 5), ( 0, 2), ( 2, 2), ( 2, 0), (10, 0), (10, 2), (12, 2), (12,10), (10,10), (10,12),( 3,12)
            ,( 3,14), ( 9,14), (14, 9), (14, 3), ( 9,-2), (3, -2), (-2, 3), (-2, 5)
            }
    	\node[cblack] (p\x) at \pt {};


    \begin{pgfonlayer}{background}
        \draw[dotted] \pathborder{p1,p2,p3,p4,p5,p6,p7,p8,p9,p10,p11,p12,p13,p14,p15,p16,p17,p18,p19}{1cm};
        \draw ($(p1)$) \foreach \x in {2,...,11} {--($(p\x)$)};
        \draw[thick, black, dashed] (s1)--(s2)--(s5)--(s6)--(s8);
        \draw[thick, black, dashed] (s1)--(s3)--(s4)--(s7)--(s8)--(s9)--(s10)--(s7)--(s6)--(s11)--(s12)--(s5)--(s4)--(s13)--(s14)--(s3)--(s2)--(s15)--(s16);

    \end{pgfonlayer}

    \node[cblue] () at ($(p4)!.5!(p5)$) {};
    \node[cblue] () at ($(p7)!.5!(p8)$) {};


    \node[white,draw,scale=.001,circle] (b1) at (-3, -3) {};
    \node[white,draw,scale=.001,circle] (b2) at (-3, 15) {};
    \node[white,draw,scale=.001,circle] (b3) at (15, 15) {};
    \node[white,draw,scale=.001,circle] (b4) at (15, -3) {};
\end{tikzpicture} 
        \caption{All except one of the two interior points are covered, regardless of which path is chosen.
        %
        }
        \label{fig:separation3}
    \end{subfigure}
    \caption{The separation gadget, step by step.}
\end{figure*}

\section{The Reduction}
Let $\Phi$ be a formula given as input to the (3,B2)-SAT problem.  We construct a polygonal curve $P$ and a point set $S$ such that $\Phi$ is satisfiable if and only if there exists polygonal curve $Q$ whose vertices are exactly $S$ with \Frechet distance less than $\eps$ from $P$.

First, we construct a gadget consisting of components of $P$ and $S$ that will force any algorithm to choose between two possible polygonal path constructions.  The gadget is constructed in such a way that these two choices are the only possible polygonal paths along the gadget's component of $S$ with \Frechet distance less than $\eps$ from $P$.  These two path possibilities will correspond to TRUE and FALSE assignments for a given variable.

Then, we create a series of points in $S$ to represent the clauses in $\Phi$, one point for each clause.  For each variable, a gadget will be placed so that the pair clause points representing the clauses in which the variable's positive instances occur are only reachable along one of the two curve possibilities, and likewise for the negative instances.  Once this has been done for each variable in $\Phi$, any polygonal curve $Q$ who vertices are exactly $S$ with $\delta_F(P, Q) \le \eps$ will correspond to an assignment to the variables of $\Phi$ in which every clause is satisfied, thus making the formula evaluate to TRUE.  Furthermore, if no such curve exists, then there can be no such satisfying assignment for $\Phi$.

\subsection{Separation Gadget}
We begin the description of our main gadget with a specific example, which we will proceed to generalize.  Consider the problem instance shown in Figure \ref{fig:separation1}, with $S = \{a, b, c, d\}$.  It is clear that the answer to this problem instance is ``no''; no polygonal curve on $S$ with $\delta_F(P, Q) \le \eps$ can visit both $b$ and $c$.  However, suppose this $P$ and $S$ were part of a larger problem instance.  Suppose further that other segments of $P$ come within $\eps$ of $b$ and $c$.  The answer to the problem instance is no longer so obvious.  Even if both points cannot be reached the first time they are encountered, it is possible that whichever point was skipped could be covered in the future.
This creates the fundamental difficulty that leads to our reduction.

Figure \ref{fig:separation2} shows an extension of the previous configuration, with more corners, all symmetrically the same as the first.  Note how we have not increased the number of options; there are still only two possible paths to take.  We can add as many of these corners as we like without breaking this property, as long as they all bend in the same direction.

The corner points must be placed very precisely to ensure the above properties hold.  Because their position is so constrained, using them to represent elements of $\Phi$ in our construction would be difficult.  At each corner, the two path possibilities alternate between the boundary of the cylinders and the interior.  As shown in Figure \ref{fig:separation2}, extra points in the cylinder interior are still only visible from the other interior points, and therefore we can add as many as we like without affecting the path possibilities.  Thus, by extending the segments between the corners, we can create large regions which are only reachable along one of the two possibilities.

There is still a problem to be addressed; as more and more corners are added, more and more points are created that would be skipped by the chosen path.  We would like to create a construction that forces a choice between \emph{only} the points in the cylinder interiors, and ensure all the corner points will be visited regardless of which path is chosen.  To accomplish this, after the last corner of the gadget, we can have $P$ loop back around along the outer edge, covering all the corner points without covering any of the interior points.  Figure \ref{fig:separation3} demonstrates this configuration.

Figure \ref{fig:partial} shows an example usage of the gadget.  The points in the cylinder interiors represent the clauses in which the variable appears.  Only one set of clause points, either the clauses in which the positive literals appear or the clauses in which the negative literals appear, can be reached.  However, all the corner points will always be covered. The full construction will include one of these gadgets for every variable, with each one passing through the points corresponding to clauses containing the variable's positive and negative literals.


\subsection{$\alpha$-Corners}
We now give the full specifications of the corner constructs, granting them the flexibility to bend at an arbitrary angle $\alpha$.  Each corner consists of two components of $P$ with four and three segments respectively, as well as four points of $S$.

Figure \ref{fig:corner} shows the full $\alpha$-corner construction.  The two components of $P$ are $A, B, C, D, E$, which we refer to as the \emph{forward path}, and $F, G, H, I$, which we refer to as the \emph{return path}.  The four points of $S$ are $G, H, K$, and $L$.  The line lengths are chosen so that the following properties hold:
\begin{itemize}
\item $\norm{\overline{BL}} = \norm{\overline{DK}} = \eps$
\item $\overline{AB}$ (resp. $\overline{DE}$) and the parallel line through $K$ (resp. $L$) are separated by exactly $\eps$
\item $G, H, K$, and $L$ are all collinear
\item The line through midpoints of $\overline{AJ}$ and $\overline{EF}$ passes through $C$.
\end{itemize}

The last property is enforced so that any points of $S$ in cylinders before $\overline{AB}$ are not visible to any points in cylinders after $\overline{DE}$, thus ensuring that external points cannot break the properties of the corner.  The compactness of $\alpha$-corners can be shown using simple geometry; an infinite strip along $\overline{AB}$ of thickness $14\eps$ is sufficient to contain all points of the structure for any $\alpha$.

\begin{figure}
    \centering
    \begin{tikzpicture}[scale=2.5,rotate=0]
    \tikzdefines

    \corner{a1}{-1,0}{0,0}{0,0}{1,1}

    \node[cblack, label=90:{\small $A$}] (a1_aa) at ($(a1_a) + ($(0,0)!1*((3+cos(45))/sin(45))*\ceps!(-5,0)$)$) {};
    \node[cblack, label=90:{\small $E$}] (a1_cc) at ($(a1_c) + ($(0,0)!1*((3+cos(45))/sin(45))*\ceps!(5,5)$)$) {};
    \node[cblack, label=below:{\small $J$}] (a1_gg) at ($(a1_g) + ($(0,0)!1*(2/sin(45))*\ceps!(-5,0)$)$) {};
    \node[cblack, label=below:{\small $F$}] (a1_ff) at ($(a1_f) + ($(0,0)!1*(2/sin(45))*\ceps!(5,5)$)$) {};
    \node[cblack, label={[label distance=0cm]-135:{\small $B$}}] () at (a1_a) {};
    \node[cblack, label=90:{\small $C$}] () at (a1_b) {};
    \node[cblack, label={[label distance=0cm]0:{\small $D$}}] () at (a1_c) {};
    \node[draw=none, label=below:{\small $L$}] at (a1_d) {};
    \node[draw=none, label=below:{\small $K$}] at (a1_e) {};
    \node[draw=none, label=below:{\small $G$}] at (a1_f) {};
    \node[draw=none, label={[label distance=-.1cm]-90:{\small $H$}}] at (a1_g) {};

    \begin{pgfonlayer}{background}
        \draw ($(a1_a)!.1cm!(a1_aa)$) arc (180:45:.1cm);
        \draw ($(a1_b)!.1cm!(a1_a)$) arc (-135:0:.1cm);
        \draw ($(a1_c)!.1cm!(a1_b)$) arc (180:45:.1cm);
        \draw ($(a1_ng)!.09cm!(a1_gg)$) arc (180:22.5:.09cm);
        \draw ($(a1_nf)!.09cm!(a1_ff)$) arc (45:202.5:.09cm);
        \draw ($(a1_ng)!.11cm!(a1_gg)$) arc (180:22.5:.11cm);
        \draw ($(a1_nf)!.11cm!(a1_ff)$) arc (45:202.5:.11cm);



        \draw[gray] ($(a1_aa)!.5!(a1_gg)$)--(a1_b)--($(a1_cc)!.5!(a1_ff)$);
        \draw[gray] (a1_nd) -- ($(a1_cc)!-.5!(a1_ff)$);
        \draw[gray] (a1_ne) -- ($(a1_aa)!-.5!(a1_gg)$);

    \end{pgfonlayer}

    \draw[thick] (a1_aa) -- (a1_a) -- (a1_b) -- (a1_c) -- (a1_cc);
    \draw[thick] (a1_ff) -- (a1_nf) -- (a1_ne) -- (a1_nd) -- (a1_ng) -- (a1_gg);
\end{tikzpicture} 
    \caption{An $\alpha$-corner for $\alpha = \pi/4$, with various properties highlighted. $\angle ABC = \pi - \alpha$, $\angle JHL = \pi - \alpha/2$}
    \label{fig:corner}
\end{figure}

\begin{figure*}[ht]
    \begin{subfigure}[t]{0.45\textwidth}
        \begin{tikzpicture}[x=1cm,y=1cm,scale=.73]
    \tikzdefines
    \foreach \x in {0,...,11}
        \node[cblue] (c\x) at (\x*30:3) {};
    \node[label={left:$x$}] () at (c3) {};
    \node[label={left:$x$}] () at (c10) {};
    \node[label={above:$\overline{x}$}] () at (c7) {};
    \node[label={above:$\overline{x}$}] () at (c1) {};

    \coordinate (s0) at (0, -4);
    \coordinate (t0) at (2, -4);
    \coordinate (s1) at (c10);
    \coordinate (t1) at (c3);
    \coordinate (s2) at (-3, 4.5);
    \coordinate (t2) at (-4, 4.5);
    \coordinate (s3) at (-4.5, 1);
    \coordinate (t3) at (-4.5, 0);
    \coordinate (s4) at (c7);
    \coordinate (t4) at (c1);

    \coordinate (s5) at (4, -1);
    \coordinate (t5) at (4, 1);
    \coordinate (s6) at (c0);
    \coordinate (t6) at (c7);
    \coordinate (s7) at (-5, -3);
    \coordinate (t7) at (-5, -4);
    \coordinate (s8) at (-5, -4.5);
    \coordinate (t8) at (-3, -4.5);
    \coordinate (s9) at (c7);
    \coordinate (t9) at (c5);

    \foreach \x [remember=\x as \y (initially 0)] in {1,...,4} {
        \corner{a\x}{s\y}{t\y}{s\x}{t\x}
    }
    \cornercap{a5}{s4}{t4}{2cm}


    \node[cblue] (sp1) at ($(s0)!\ceps!90:(t0)$) {};

    \begin{pgfonlayer}{background}
        \draw[thick]
            (s0) node[cblack] {}
            \foreach \x in {1,...,5} { --(a\x_a)--(a\x_b)--(a\x_c) }
            \foreach \x in {4,...,1} { --(a\x_f)--(a\x_g) }
            -- (a1_g -| s5) coordinate(f1);
        \draw[gray, dotted] \pathborder{s0,a1_a,a1_b,a1_c,a2_a,a2_b,a2_c,a3_a,a3_b,a3_c,a4_a,a4_b,a4_c,a5_a,a5_b,a5_c,a4_f,a4_g,a3_f,a3_g,a2_f,a2_g,a1_f,a1_g,f1}{\ceps}; 

    \end{pgfonlayer}

    \node[cblue] () at (f1) {};
    \node[cblack] () at (f1) {};



\end{tikzpicture} 
        \bigskip
        \caption{A partial construction for formula with 12 clauses, showing the gadget for a single variable.  The only two valid paths visit the either positive literal's clauses or the negative literal's clauses.  }
        \label{fig:partial}
    \end{subfigure}
    \hfill
    \begin{subfigure}[t]{0.45\textwidth}
        \begin{tikzpicture}[scale=.73]
    \tikzdefines
    \node[cblue, label={[label distance=.3cm]95:$1$}] (c0) at (0*90+45:3) {};
    \node[cblue, label={[label distance=.3cm]85:$2$}] (c1) at (1*90+45:3) {};
    \node[cblue, label={[label distance=.3cm]-85:$3$}] (c2) at (2*90+45:3) {};
    \node[cblue, label={[label distance=.3cm]-95:$4$}] (c3) at (3*90+45:3) {};

    \coordinate (s0) at (0, -4);
    \coordinate (t0) at (1, -4);
    \coordinate (s1) at (c3);
    \coordinate (t1) at (c0);
    \coordinate (s2) at (2, 4.5);
    \coordinate (t2) at (0, 5.5);
    \coordinate (s3) at (0, 5.5);
    \coordinate (t3) at (-2, 4.5);
    \coordinate (s4) at (c1);
    \coordinate (t4) at (c2);

    \coordinate (s5) at (4, -2.9);
    \coordinate (t5) at (4, 1);
    \coordinate (s6) at (c0);
    \coordinate (t6) at (c1);
    \coordinate (s7) at (-4.5, 2);
    \coordinate (t7) at (-5.5, 0);
    \coordinate (s8) at (-5.5, 0);
    \coordinate (t8) at (-4.5, -2);
    \coordinate (s9) at (c2);
    \coordinate (t9) at (c3);

    \coordinate (s10) at (5, 3);
    \coordinate (t10) at (4, 4);
    \coordinate (s11) at (c0);
    \coordinate (t11) at (c2);
    \coordinate (s12) at (-3, -5);
    \coordinate (t12) at (0, -6);
    \coordinate (s13) at (0, -6);
    \coordinate (t13) at (3, -5);
    \coordinate (s14) at (c3);
    \coordinate (t14) at (c1);

    \foreach \x [remember=\x as \y (initially 0)] in {1,...,4} {
        \corner{a\x}{s\y}{t\y}{s\x}{t\x}
    }
    \cornercap{a5}{s4}{t4}{2cm}
    \foreach \x [remember=\x as \y (initially 5)] in {6,...,9} {
        \corner{a\x}{s\y}{t\y}{s\x}{t\x}
    }
    \cornercap{a10}{s9}{t9}{2.75cm}
    \foreach \x [remember=\x as \y (initially 10)] in {11,...,14} {
        \corner{a\x}{s\y}{t\y}{s\x}{t\x}
    }
    \cornercap{a15}{s14}{t14}{3cm}

    \node[cblue] (sp1) at ($(s0)!\ceps!90:(t0)$) {};
    \node[cblue] (sp2) at ($(t5)!\ceps!270:(s5)$) {};
    \node[cblue] (sp3) at ($(s10)!\ceps!90:(t10)$) {};

    \begin{pgfonlayer}{background}
        \draw[thick]
            (s0)
            \foreach \x in {1,...,5} { --(a\x_a)--(a\x_b)--(a\x_c) }
            \foreach \x in {4,...,1} { --(a\x_f)--(a\x_g) }
            -- ++(1,0) coordinate(f1) {}
            -- (s5) coordinate(f2) {}
            \foreach \x in {6,...,10} { --(a\x_a)--(a\x_b)--(a\x_c) }
            \foreach \x in {9,...,6} { --(a\x_f)--(a\x_g) }
            -- (a6_g -| s10) coordinate(f3) {}
            -- (s10)
            \foreach \x in {11,...,15} { --(a\x_a)--(a\x_b)--(a\x_c) }
            \foreach \x in {14,...,11} { --(a\x_f)--(a\x_g) }
            ;
        \draw[gray, dotted] \pathborder{s0,a1_a,a1_b,a1_c,a2_a,a2_b,a2_c,a3_a,a3_b,a3_c,a4_a,a4_b,a4_c,a5_a,a5_b,a5_c,a4_f,a4_g,a3_f,a3_g,a2_f,a2_g,a1_f,a1_g,f1,f2,a6_a,a6_b,a6_c,a7_a,a7_b,a7_c,a8_a,a8_b,a8_c,a9_a,a9_b,a9_c,a10_a,a10_b,a10_c,a9_f,a9_g,a8_f,a8_g,a7_f,a7_g,a6_f,a6_g,f3,s10,a11_a,a11_b,a11_c,a12_a,a12_b,a12_c,a13_a,a13_b,a13_c,a14_a,a14_b,a14_c,a15_a,a15_b,a15_c,a14_f,a14_g,a13_f,a13_g,a12_f,a12_g,a11_f,a11_g}{\ceps};

        \end{pgfonlayer}

    \node[cblue] () at (f1) {};
    \node[cblack] () at (f1) {};
    \node[cblue] () at (f2) {};
    \node[cblack] () at (f2) {};
    \node[cblue] () at (f3) {};
    \node[cblack] () at (f3) {};


\end{tikzpicture} 
        \caption{A completed construction for the formula
            $\Phi = (x \vee y \vee z) \wedge
            (\overline{x} \vee y \vee \overline{z}) \wedge
            (\overline{x} \vee \overline{y} \vee z) \wedge
            (x \vee \overline{y} \vee \overline{z})$.
            The upper right clause point represents the first clause, and the second, third, and fourth follow counterclockwise.
            }
        \label{fig:complete}
    \end{subfigure}
    \caption{Example constructions}
\end{figure*}
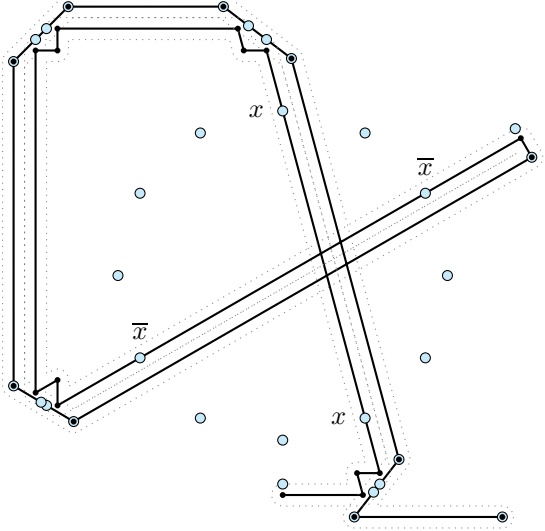
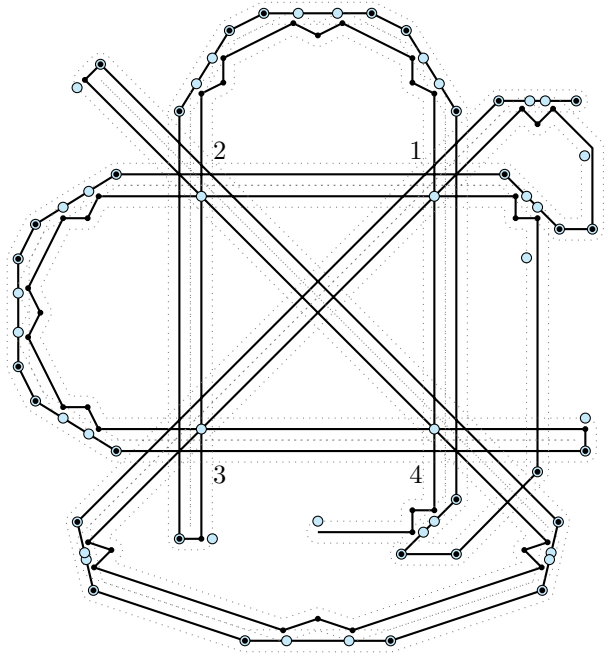

\subsection{Construction}
Using this gadget, we can now build the final construction.

We begin by adding an initial set of points to $S$ which we refer to as ``clause points'' $C_1,\dots,C_n$, one for each clause in $\Phi$.  We position these points so that they lie, equally spaced, on a circle of radius $n^2 \eps$.  This radius ensures that, if the lines $\overline{C_i C_j}$ and $\overline{C_k C_l}$ are parallel, they are separated by at a distance of least $14\eps$.  Let the $\emph{clause strip}$ about clauses $i$ and $j$ denote the set of points within $7\eps$ of the line $\overline{C_i C_j}$.

We now proceed to create $P$, adding more points to $S$ as needed.  As we place components of $P$, we require that all joints and $\alpha$-corners be placed entirely outside all strips about all clause pairs, so as not to block future pieces.  Furthermore, we require that new pieces of $P$ be placed so that their cylinders do not intersect the convex hull of all previously placed points in $S$.  This ensures that previously placed pieces do not create unintended ``shortcuts'' that could break the properties of the new pieces.  The exception to both these rules is the segment that passes directly through two clause points $i$ and $j$, which, of course, must pierce the convex hull.  Its adjacent $\alpha$-corners will lie entirely inside the strip about $\overline{C_i C_j}$, but must be placed outside all other clause strips.  Note that this is always possible; beyond $4 n^2  \eps$ units from the center of the clause ring, no clause strip intersects any other.  Strips of different angles will grow further and further apart, creating regions of arbitrary size between them.

So long as the requirements in the preceding paragraph are met, the start point of $P$ can be placed arbitrarily.  We then perform the following procedure for each variable $v_i$ in $\Phi$, building the construction incrementally.  Let $x$ and $y$ be the clauses in which the positive literals of $v_i$ occur, and $z$ and $w$ be the clauses in which the negative literals occur.  We begin by positioning an $\alpha$-corner so that the extension of the last segment of the forward path passes through $C_x$ and $C_y$.  An extra point, which we refer to as a split point, is added on the boundary of the first forward path segment to induce the splitting of the two possible paths.  From there, both the forward and return paths are extended through the clause ring, with the forward path crossing through $C_x$ and $C_y$.

On the opposite side, outside the convex hull of all points in $S$ so far, another $\alpha$-corner is added, bending the path toward $\overline{C_z C_w}$.  More $\alpha$-corners, all bending in the same direction, are added as needed until one can be placed such that the forward path passes through $C_z$ and $C_w$.  Note that there must be an odd number of $\alpha$-corners in order to ensure that $C_x$, $C_y$ and $C_w$, $C_z$ are reachable on different curve possibilities.  Once the paths have been extended through the clause ring and outside the convex hull, another split point is added on the boundary to collapse the curve possibilities.  Finally, the forward path is linked to the return path, and the joint is added to $S$.  At the end of the return path, more segments of $P$ are added, with each joint being added to $S$, in order to move to the next variable's clause strips.

Once this process has been completed for all variables, the construction is complete.  Note the units in our construction are all in terms of $\eps$, so $\eps$ can be chosen arbitrarily.  Figure \ref{fig:complete} shows a completed construction for a very simple formula, while Figure \ref{fig:partial} shows a partial construction for a more complex formula.

\section{Result}

\begin{lemma}
There exists a polygonal path $Q$ on $S$ with $\delta_F(P, Q) \le \eps$ that visits every point in $S$ if and only if $\Phi$ is satisfiable.
\end{lemma}
\begin{proof}
For the forward direction, assume $\Phi$ has a satisfying assignment.  It is easy to see that our construction always has a polygonal path $Q$ on $S$ with $\delta_F(P, Q) \le \eps$ that will visit every non-clause point; $\alpha$-corners are constructed specifically to ensure this.  If $\Phi$ has a satisfying assignment, then one of the two path possibilities in each variable construct will cover the the clause points corresponding to the clauses satisfied by that variable, resulting in all clause points being visited as well.

For the backward direction, let $Q$ be a complete polygonal path $Q$ on $S$ with $\delta_F(P, Q) \le \eps$.  By constructing each variable construct completely outside the convex hull of all previously placed points of $S$, we have ensured that any $Q$ with $\delta_F(P, Q) \le \eps$ must follow the path we have laid out.  Each variable construct forces a choice between two paths, representing a true or false value for that variable.  Since each $Q$ visits each clause point, the path taken in each variable construct represents an assignment to the variables that satisfies $\Phi$.
\end{proof}

It is straightforward to show that five $\alpha$-corners is sufficient to move between any two strips.  Thus, the construction is clearly of polynomial size.  This, together with the fact that the problem is in NP, leads to the final result.
\begin{theorem}
The Non-unique All-points Continuous CPSM Problem is NP-complete.
\end{theorem}

In the construction, the only points that occur more than once are the clause points and the inner $\alpha$-corner points.  In all occurrences of both cases, the next point is always reachable from the previous point.  Thus, for this class of problem instances, any solution to the Non-unique version of this problem can be converted to a solution to the Unique version by simply skipping the points that have already been visited.  This shows that the same reduction applies to the Unique version.

\begin{corollary}
The Unique All-points Continuous CPSM Problem is NP-complete.
\end{corollary}

\section{Restricted Problem}
The hardness of the problem stems from the fact that $P$ could come within $\eps$ of a point in $S$ multiple times.  Thus, it is natural to ask if the hardness remains if we make a restriction that prevents such a situation.  In the following section, we show that the Non-unique All-Points Continuous CPSM problem is polynomial-time solvable under the condition that, for all $s \in S$, the set $\{t \in [0, 1] \mid \norm{P(t), s} \le \eps\}$ is connected.

\subsection{Algorithm Outline}
Let $P_i$ be the $i$th segment of $P$, and let $C_i$ be the cylinder of radius $\eps$ around $P_i$.  For simplicity, let $C_0 = \mathcal{B}(P(0), \eps)$ and $C_{n+1} = \mathcal{B}(P(1), \eps)$.  Let $S_i = C_i \cap S$.  For a point $s \in S$, let $l(s)$ be the earliest occurring point of $P$ that is within $\eps$ of $s$, and let $r(s)$ be the latest.  By the restriction imposed in the preceding section, $l(s)$ and $r(s)$ are uniquely defined.

A obvious preprocessing step is to confirm that all points of $S$ are in some cylinder $C_i$.  Another is to confirm that $S_0$ and $S_{n+1}$ are nonempty.  Since instances that do not satisfy these properties can be immediately ruled out, we assume them to be true for the remainder of this section.

The $i$th and $j$th segments of $P$ are said to be \emph{connectable} via $(s, t)$, where $s \in S_i$ and a $t \in S_j$, when the following properties hold.  Note that $s$ and $t$ could be the same point.
\begin{enumerate}
    \item $\bigcup_{i < k < j} S_k - S_i - S_j = \emptyset$
    \item $\delta_F(P', \overrightarrow{st}) \le \eps$, where $P'$ is the subcurve of $P$ from $r(s)$ to $l(t)$
    \item $\forall \, v \in S_i \cup S_j$, $l(v) \preceq r(s)$ or $l(t)\preceq r(v)$
\end{enumerate}

Let $G$ be a directed graph whose vertices correspond to the segments of $P$, with an edge between any two connectable segments of increasing index.

\begin{theorem} \label{bigiff}
    There is a path in $G$ from $P_1$ to $P_n$ if and only if there exists a polygonal curve $Q$ whose vertices are exactly $S$ with $\delta_F(P, Q)$.
\end{theorem}

\begin{proof} $(\Rightarrow)$
To show the forward direction of this statement, we will construct a curve with the requisite properties under the assumption that $G$ has such a path, denoted by $1 = a_1 < a_2 < \dots < a_m = n$.  Let $t_i$ and $s_i$ be the points in $S_{a_i}$ connecting $P_{a_i}$ to $P_{a_{i-1}}$ and $P_{a_{i+1}}$ respectively.  In other words, if $C_{a_{i-1}}$ and $C_{a_i}$ are connectable via $(x, y)$, and $C_{a_i}$ and $C_{a_{i+1}}$ are connectable via $(z, w)$, then $t_i = y$ and $s_i = z$.  Let $t_1$ (resp. $s_m$) be an arbitrary point in $S_0$ (resp. $S_{n+1}$); these points will be the first and last points of $Q$.

Repeat the following for each $i$ from 1 to $m$.  Add $t_i$ to $Q$.  Then, visit every point in the set $\{v \in S_{a_i} \mid l(t_i) \preceq r(v)$ and $l(v) \preceq r(s_i)\} - \{s_i, t_i\}$, in order monotonic along the direction of $P_{a_i}$.  Finally, add $s_i$ to $Q$.  Once this process has been performed for each $i$, the path is complete.

We must show now that the constructed curve (i) visits every point in $S$ and (ii) has \Frechet distance at most $\eps$ from $P$.

We prove (i) first.  For a given iteration $i$ of the above process, assume some point in $v \in S_{a_i}$ where \emph{not} did not satisfy the criteria $l(t_i) \preceq r(v)$ and $l(v) \preceq r(s_i)$.  Without loss of generality, assume $v$ did not satisfy the former of the two conditions.  Then by Property 3, $l(v) \preceq r(s_{i-1})$.  If $l(t_{i-1}) \preceq r(v)$ as well, then $v$ would have been added to $Q$ in the $i-1$ iteration.  If not, then $l(v) \preceq r(s_{i-2})$, again by Property 3.  Since $l(t_1) \preceq r(v)$ for all $v \in S$, it follows that $v$ must have been added to $Q$ during some previous iteration.  An identical argument can be used to show that, if $v$ had failed the latter condition, $v$ would definitely be added during some future iteration.  Thus, the final constructed curve contains every point in $\bigcup_{i \in [1, m]} S_{a_i}$, which by Property 1 above, equals $S$.

Now we prove (ii).  To accomplish this, we need only show that $\delta_F(P', Q') \le \eps$, where $P'$ the subcurve of $P$ from $l(t_i)$ to $r(s_i)$ and $Q'$ is the subcurve of $Q$ consisting of all points added during the $i$th iteration.  Combining this result with Property 2 will prove the statement by concatenation.

Fix an iteration $i$, and let the sequence $t_i = v_1, v_2, \dots, v_m = s_i$ denote the sequence of points added during the $i$th iteration.  For $j \in [2, m-1]$, let $p_j$ denote the point on the segment $P' \cap P_{a_i}$ closest to $v_j$.  Let $p_1 = l(t_i)$ and $p_m = r(s_i)$.  It follows from Property 3 and from the fact all points other than $t_i$ and $s_i$ are added to $Q$ in order monotonic along the direction of $P_{a_i}$ that $p_1 \preceq p_2 \preceq \dots \preceq p_m$.  Thus, the concatenation of $p_1$ through $p_m$ is exactly $P'$.  Since $\norm{v_j, p_j} \le \eps$ for all $j$, either by definition or by virtue of being in $C_{a_i}$, it follows by concatenation that $\delta_F(P', Q') \le \eps$.
\end{proof}

\begin{proof} $(\Leftarrow)$
We now prove the backward direction of Theorem \ref{bigiff}.  Assume there exists a polygonal curve $Q$ whose vertices are exactly $S$ with $\delta_F(P, Q) \le \eps$.  Then there exists a sequence of edges in $Q$ that correspond to a path in $G$ from $P_1$ to $P_n$.

Let $q_1, \dots, q_m$ be the sequence of vertices of $Q$.  From the definition of \Frechet distance, there exist continuous non-decreasing surjective functions $\sigma$ and $\tau$ such that $\max_{t\in[0,1]} \norm{P(\sigma(t)), Q(\tau(t))} \le \eps$.  Let $t_i$ be the minimum value such that $Q(\tau(t_i)) = q_i$.  As a special case, we impose that $t_m = 1$.  Then, let $a_i$ be the index of the segment on which $P(\sigma(t_i))$ lies.  Note that the sequence $A = a_1, \dots, a_m$ is nondecreasing. We claim that, if $a_i < a_{i+1}$, then $P_{a_i}$ and $P_{a_{i+1}}$ are connectable via $(q_i, q_{i+1})$. Thus, the maximal strictly increasing subsequence of $A$ corresponds to a path in $G$.

To show this, let $i$ be such that $a_i < a_{i+1}$  Since $P(\sigma(t_i)) \in P_{a_i}$ and $Q(\tau(t_i)) = q_i$ are within $\eps$ of each other, it must be that $q_i \in S_{a_i}$, and likewise for $i-1$.  We now show that the three properties of connectable segments are satisfied.  The second property is immediate, as it is satisfied for all adjacent vertices of $Q$.

For the first property, assume the set $\bigcup_{a_{i} < k < a_{i+1}} S_k - S_{a_{i}} - S_{a_{i+1}}$ were nonempty, and let $q_j$ be a point in this set.  Since $q_j$ is outside the cylinders $C_{a_i}$ and $C_{a_{i+1}}$, $q_j$ must be separated by more than $\eps$ from either segment.  Furthermore, by our restriction on the input, $q_j$ cannot be in any cylinder before $C_{a_i}$ or after $C_{a_{i+1}}$.  Thus, it must be that $a_{i} < a_j < a_{i+1}$.  This contradicts the fact that $A$ is nondecreasing.

For the third property, assume there was some point $q_j \in S_i \cup S_j$ such that $r(q_i) \prec l(q_j)$ and $r(q_j) \prec l(q_{i+1})$.  The former condition implies that $q_j$ must be visited after $q_i$, and the latter condition implies that $q_j$ must be visited before $q_{i+1}$.  This is a contradiction.
\end{proof}

With this property in hand, a polynomial time algorithm follows by simply finding the connecting edges, constructing the graph $G$, and then constructing $Q$ as in the proof above if $G$ has a valid path.

\section{Conclusion and Open Problems}
We have shown that both the Unique and Non-unique versions of the All-Points Continuous CPSM problem are NP-complete.  Furthermore, for Non-unique case, we have shown that a modified version of the problem can be solved in polynomial time.

The Unique Subset Continuous version remains an open problem.  However, given that the Discrete version is NP-complete, and that the Continuous versions tend to be harder, this version is almost certainly NP-complete as well.

{
\bibliographystyle{plain}
\bibliography{frechet}
}
\end{document}